\documentclass[journal,draftcls,onecolumn,12pt]{IEEEtran}
\usepackage{amssymb}
\normalsize
\usepackage[final]{graphicx}
\usepackage{fancybox,epsfig,epsf,amssymb,subfigure,latexsym,times,cite,array,color,amsthm,amsmath,psfrag}
\def\bx{{\bf x}}
\def\by{{\bf y}}
\def\bg{{\bf g}}
\def\bA{{\bf A}}
\def\bI{{\bf I}}
\def\bn{{\bf n}}
\def\Exp{\mathbb{E}}
\allowdisplaybreaks
\def\bG{{\bf G}}
\newtheorem{theorem}{Theorem}
\newtheorem{Lemma}{Lemma}
\newtheorem{Proposition}{Proposition}

\begin{document}

\title{Performance Analysis and Location Optimization for Massive MIMO Systems with Circularly Distributed Antennas}

\author{Ang Yang,
Yindi Jing,~\IEEEmembership{Member,~IEEE,} Chengwen Xing,~\IEEEmembership{Member,~IEEE,} \\
Zesong
Fei,~\IEEEmembership{Member,~IEEE,} Jingming Kuang
\thanks{A.~Yang, C.~Xing, Z.~Fei, and J.~Kuang are with the School of Information and Electronics,
Beijing Institute of Technology, Beijing, China.
(Email: taylorkingyang@163.com, chengwenxing@ieee.org,
feizesong@bit.edu.cn, JMKuang@bit.edu.cn).}

\thanks{Y.~Jing is with the Department of Electrical and Computer Engineering,
University of Alberta, Edmonton, AB, T6G 2V4 Canada (e-mail:
yindi@ualberta.ca).}

}

\maketitle
\vspace{-4mm}
\begin{abstract}
In this paper, we analyze the achievable rate of the uplink of a
single-cell multi-user distributed massive
multiple-input-multiple-output (MIMO) system. The multiple users are
equipped with single antenna and the base station (BS) is equipped
with a large number of distributed antennas. We derive an analytical
expression for the asymptotic ergodic achievable rate of the system
under zero-forcing (ZF) detector. In particular, we consider
circular antenna array, where the distributed BS antennas are
located evenly on a circle, and derive an analytical expression and
closed-form tight bounds for the achievable rate of an arbitrarily
located user. Subsequently, closed-form bounds on the average
achievable rate per user are obtained under the assumption that the
users are uniformly located in the cell. Based on the bounds, we can
understand the behavior of the system rate with respect to different
parameters and find the optimal location of the circular BS antenna
array that maximizes the average rate. Numerical results are
provided to assess our analytical results and examine the impact of
the number and the location of the BS antennas, the transmit power,
and the path-loss exponent on system performance. It is shown that
circularly distributed massive MIMO system largely outperforms
centralized massive MIMO system.

\end{abstract}





\IEEEpeerreviewmaketitle

\section{Introduction}
With the demands of the wireless data services nowadays, high
spectrum efficiency or data rate is undoubtedly an important feature
of future wireless systems
\cite{spectral_efficiency_1,spectral_efficiency_2}. In order to
improve the data rate of wireless systems, various innovative ideas
have been proposed and investigated. Among the most successful ones
in recently years is the multiple-input-multiple-output (MIMO)
concept \cite{MIMO_1,MIMO_2}.

Recently, distributed MIMO systems, or distributed antenna systems,
was proposed to further improve the data rate, in which multiple
transmit or receive antennas are distributively located to reduce
the physical transmission distance between the transmitter and the
receiver
\cite{Outage_performance,Distributed_Antenna,Honglinhu_2007,TWU_WCNC,Huiling_Zhu_Resource,DAS_tsinghua_1,huilin_zhu_1,Antenna_placement,huilin_zhu_2,DAS_tsinghua_2}.
The distributive antennas are assumed to be connected to the central
unit via high-bandwidth and low-delay backhaul such as optical fiber
channels. It has been proved that distributed MIMO outperforms
centralized MIMO in outage probability and achievable rate. In
\cite{Outage_performance}, for a single-cell single-user distribute
MIMO system with arbitrary antenna topology, the authors analyzed
the outage performance as well as the diversity and multiplexing
gains. For a single-cell distributed MIMO systems with multiple
uniformly distributed users, upper bounds on the ergodic capacity of
one user and approximate expressions of sum capacity of the cell
were derived in \cite{DAS_tsinghua_1}. In \cite{Antenna_placement},
both single-cell and two-cell distributed MIMO systems with multiple
uniformly distributed users per cell are considered. The cells are
assumed to be circular and the distributed base station (BS)
antennas have circular layout. The locations of the distributive
antennas were optimized to maximize lower bounds on the expected
signal to noise ratio (SNR) and signal to leakage ratio. The
resource (including power, subcarrier, and bit) allocation problems
in single-cell multi-user distributed antenna systems were
investigated in \cite{huilin_zhu_2}. In \cite{TWU_WCNC}, different
radio resource management schemes were compared for multi-cell
multi-user distributed antenna systems. In \cite{DAS_tsinghua_2},
for multi-cell networks with multiple remote antennas and one
multi-antenna user in each cell,  the input covariances for all
users were jointly optimized to maximize the achievable ergodic sum
rate.

In the above literature of distributed MIMO, single user or multiple
users with orthogonal channels are assumed, so there is no
inter-user interference. However in current and future wireless
systems, it is expected to have multiple users sharing the same
time-frequency resource. In such systems, one user will suffer from
the interferences of other users in the cell, which can largely
degrade the system rate. To conquer the intra-cell user interference
problem, the concept of massive MIMO, where the BS is equipped with
a very large number of antennas usually of hundreds or higher, was
proposed and attracted considerable attention recently
\cite{spectral_efficiency_2,Marzetta,Massive_1,Massive_2,Massive_3,Massive_4,Massive_5,Massive_6,Massive_7,Massive_8,Larsson_uplink,Massive_9,Massive_10}.
With a large number of antennas at the BS, according to the law of very long vectors, transmission channels for different users are orthogonal to each other. User-interference diminishes and very high data rate can be achieved with low complexity signal processing. 
The ergodic achievable rates of the single-cell
multi-user massive MIMO uplink with linear detectors, i.e. maximum ratio
combining (MRC), zero-forcing (ZF), minimum-mean-square-error
(MMSE), have been derived in \cite{Larsson_uplink}. The achievable rates of both the uplink and downlink of multi-cell multi-user massive MIMO systems with linear precoders and detectors were analyzed in \cite{Massive_5}.

As the combination of the two promising concepts, massive MIMO and
distributed MIMO, distributed massive MIMO systems are of great
potential in fulfilling the increasing demands of next generation
communication systems
\cite{D_massive_1_shijin_massive,D_massive_2_C_Zhong_massive,D_massive_3,D_massive_4,D_massive_5}.
The authors in \cite{D_massive_1_shijin_massive} focused on an
uplink massive MIMO system consisting of multiple users and one BS
with several large scale distributed antenna sets. In this system,
the deterministic equivalence of the ergodic sum rate was derived
and an iterative waterfilling algorithm was proposed for finding the
capacity-achieving input covariance matrices.
\cite{D_massive_2_C_Zhong_massive} analyzed the sum rates of
distributed MIMO systems, in which $L$ multiple-antenna radio ports
form a virtual transmitter and jointly transmit data to a
centralized multiple-antenna BS. The capacity and spatial degrees of
freedom of a large distributed MIMO system were investigated in
\cite{D_massive_3}, where wireless users with single transmit and
receive antenna cooperate in clusters to form distributed transmit
and receive antenna arrays. In \cite{D_massive_4}, a simplistic
matched filtering scheme and a subspace projection filtering scheme
were investigated in a fixed size single-cell network, where the BS
antennas are assumed uniformly and randomly located in the cell
serving single-antenna users. In \cite{D_massive_5}, for networks
where one large-scale distributed BS with a grid antenna layout
serving single-antenna users, the energy efficiency maximization
problem was formulated under per-antenna transmit power and per-user
rate constraints.  low complexity channel-gain-based antenna
selection method and interference-based user clustering method were
proposed to improve the system energy efficiency.

In this paper, we consider the uplink of a single-cell multi-user
distributed massive MIMO system in which the BS equipped with a
large number of distributed antennas receives information from
multiple single-antenna users. The asymptotic system
achievable rate under linear ZF detector is analyzed for arbitrary but known antenna locations and for circular antenna layout.
Based on the analysis, the location of the distributed
antennas for circular antenna layout is optimized. Compared to
\cite{D_massive_1_shijin_massive,D_massive_2_C_Zhong_massive}, the topology of the distributed BS antennas used in our work is different, as the BS in
\cite{D_massive_1_shijin_massive} is the combination of several
centralized BS and the BS in \cite{D_massive_2_C_Zhong_massive} is
centralized. Compared to \cite{D_massive_3}, the transmitters in this work are
non-cooperative single-antenna users, while the users in
\cite{D_massive_3} cooperate with each other and form a virtual
transmitter with distributed antenna arrays. Thus our system model
and derivations of achievable rate are different from
\cite{D_massive_1_shijin_massive,D_massive_2_C_Zhong_massive,D_massive_3}.
Compared to \cite{D_massive_4,D_massive_5}, we focus on the
derivations of achievable rate and the location optimization of the distributive BS antennas, while \cite{D_massive_4} faces the problem
of interference control through the use of second-order channel
statistics and \cite{D_massive_5} works on the energy efficiency
maximization problem. The major contributions of this paper are
summarized as follows.

\begin{itemize}

\item We provide new results for independent but non-identically distributed  (i.n.i.d) random vectors with very high dimension (see Lemma \ref{Lemma_law_large_number}). Based on the results, an analytical expression for the asymptotic achievable rate of multi-user distributed massive MIMO systems is derived for arbitrary but known user location and antenna deployment (see Proposition \ref{capacity total}).

\item We consider a practical circular antenna layout, where antennas of the BS are located evenly in a circle. The asymptotic achievable rates for an arbitrarily located user and two closed-form tight bounds
are derived (see Theorems \ref{capacity circle theorem label} and \ref{capacity circle bounds theorem label}).
\item Furthermore, for the circular antenna layout, tight closed-form bounds on the average rate of the cell with uniform user location are obtained (see Theorem \ref{circle average capacity}). These results can be used to predict the system performance and understand its behavior with respect to the number of antennas, the location of antenna, the cell size, and the transmit power.

\item Based on the acquired tight bounds, we
derive the optimal radius of the distributed antennas for the
maximum average achievable rate (see Lemma \ref{circle average opt theorem}), which guides the fundamental and practical problem of antenna placement  for distributed MIMO systems.

\item Finally, numerical results are provided to assess our
analysis. It is shown that multi-user distributed massive MIMO is largely
superior to centralized massive MIMO in achievable uplink rate.
\end{itemize}

The remaining of the paper is organized as follows. The system model
and asymptotic achievable rate analysis of a general distributed massive MIMO system are presented in
Section II. Asymptotic achievable rate analysis of the circularly
distributed massive MIMO is present in Section III. Location
optimization of the circularly distributed antennas is provided in
Section IV. Simulation results are presented in Section V. Conclusions are drawn in Section VI. Involved proofs are included in the appendices.

Notation: Boldface lowercase letters denote vectors, while boldface
uppercase letters denote matrices. We use ${(\cdot)}^{\rm{T}}$,
$(\cdot)^{*}$ and $(\cdot)^{\rm{H}}$ to denote the transpose,
conjugate and conjugate transpose of a matrix or a vector,
respectively. For a matrix ${\bf{Z}}$, ${\rm{Tr}}({\bf{Z}})$ is its
trace. The symbol ${\bf{I}}_{M}$ denotes the $M \times M$ identity
matrix, while ${\bf{0}}_{M,N}$ denotes the $M \times N$ matrix whose
entries are zeros. The symbol ${\mathbb{E}}$ denotes the statistical
expectation operation. The symbol $\left\| {\cdot} \right\|_F$
denotes Frobenius norm of a matrix or a vector. The function
$\log_2(\cdot)$ is the base-2 logarithm and $\ln(\cdot)$ is the
natural logarithm.

\section{System model and Asymptotic Achievable Rate Analysis}

\begin{figure}[!t]
\begin{center}
\includegraphics[width=4in]{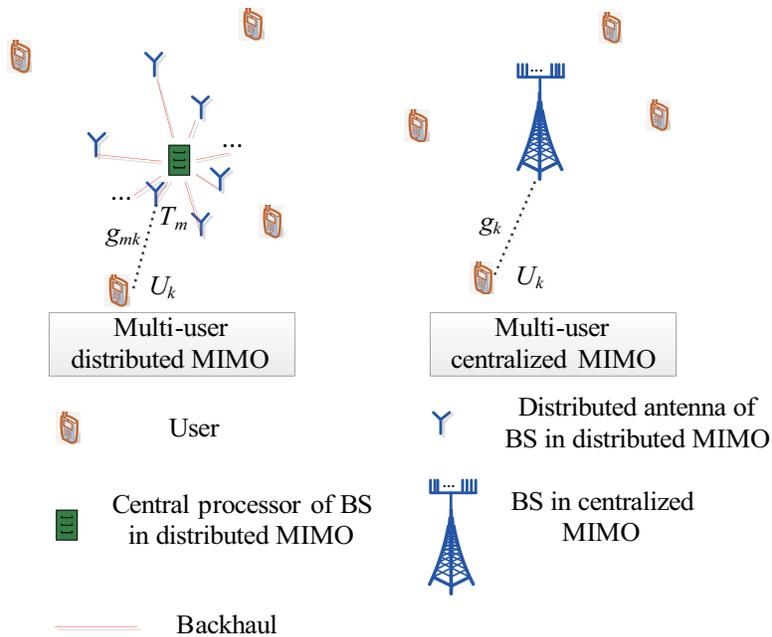}
\caption{System models of multi-user distributed MIMO (left side)
and multi-user centralized MIMO (right side).} \label{fig-1}
\end{center}
\end{figure}

\subsection{Multi-User Distributed Massive MIMO System Model}

We consider a single-cell multi-user distributed massive MIMO
system. As shown in the left side of Figure \ref{fig-1}, in this
system, there is one BS equipped with $M$ antennas which are
spatially distributed \cite{D_massive_4,D_massive_5}. The number of
antennas $M$ is assumed to be large, e.g., a few hundreds. This is
different to the multi-user centralized MIMO system
\cite{Marzetta,Massive_1,Massive_2,Massive_3,Massive_4,Massive_5,Massive_6,Massive_7,Massive_8,Larsson_uplink,Massive_9,Massive_10},
where the BS antennas are centralized and spatially co-located
(shown in the right side of Figure \ref{fig-1}). Compared with
centralized MIMO systems, distributed MIMO systems provide
macro-diversity and have enhanced network coverage and capacity, due
to their open and flexible infrastructure
\cite{Huiling_Zhu_Resource,Distributed_Antenna,D_massive_1_shijin_massive}.
We assume that the distributed BS antennas are connected with high
capacity backhaul and have ideal cooperation with each other. There
are $K$ users, each equipped with single antenna. We assume that
$M\gg K$.

Denote the channel coefficient between the $m$th antenna of the BS
and the $k$th user as $g_{mk}$. We consider both the path-loss and
small-scale fading as follows \cite{D_massive_5}
\begin{align}\label{channel matrix}
{g_{mk}} = {h_{mk}}\sqrt {{\beta _{mk}}}, \quad m = 1,2, \ldots ,M,\ k =
1,2, \ldots ,K,
\end{align}where ${h_{mk}}$ is the small-scale fading coefficient, which is model as a random variables with zero-mean and unit-variance. Without loss of generality, Rayleigh fading is adopted in the simulation, where $h_{mk}$ follows circularly symmetric complex  Gaussian (CSCG) distribution, i.e., $h_{mk}\sim\mathcal{CN}(0,1)$. $h_{mk}$'s are assumed to be mutually independent.
$\beta _{mk}$ models the path-loss. We assume that\begin{align}\label{large scale}
 {{\beta _{mk}}}  = d_{mk}^{-v},
\end{align}
where $d_{mk}$ is the distance between $m$th antenna of the BS and the $k$th user and $v$ is the path-loss exponent with typical values ranging from $2$ to $6$, i.e., $2\le \nu \le 6$. Let
\[\bg_k\triangleq\left[\begin{array}{lll} g_{1k} & \cdots & g_{Mk}\end{array}\right]^{\rm T},\] which is the $M\times 1$ channel vector between $k$th user and all $M$ BS antennas. Let
\[\bG\triangleq\left[\begin{array}{lll} \bg_{1} & \cdots & \bg_{K}\end{array}\right],\] which is the channel matrix between  all $K$ users and all $M$ BS antennas. Perfect channel state information (CSI) is assumed at the BS, that is, the BS knows $\bG$ precisely. 

The uplink communication is studied, where the $K$ users transmit their
data in the same time-frequency resource to the BS. Let $\bx$ be the $K \times 1$ signal vector containing the user data, where its $k$-th entry $x_k$ is the information symbol of the $k$th user. $\bx$ is normalized as $\Exp\{\left\|{\bf{x}}\right\|_F^2\}=1$.
Let $P$ be the average transmit power of each user. This implies that in this work, we assume that all users have the same transmit power. But the derived results can be directly employed to non-equal power case. The $M \times 1$ vector of the received
signals at the BS is
\begin{align}
\label{received signals} {\bf{y}} = \sqrt {P} {\bf{G}}{\bf{x}} +
{\bf{n,}}
\end{align}
where ${\bf{n}}$ is the noise vector, whose entries are assumed to be independent and identically distributed (i.i.d.) CSCG random variables with zero-mean and unit-variance, that is, $\bn\sim\mathcal{CN}({\bf 0},\bI_M)$.


ZF linear detector is used at the receiver, which has low-complexity
and achieves comparable sum-rate performance to other more
complicated designs such as the minimum-mean-square-error (MMSE)
detector in massive MIMO systems \cite{Larsson_uplink,Massive_5}. ZF
separates data streams from different users by multiplying the
received signal vector $\by$ with $\bA\triangleq{\left(
{{{\bf{G}}^H}{\bf{G}}} \right)^{ - 1}}{\bf{G}}^H$. From
(\ref{received signals}), we have
\begin{align}\label{detector 3}
{\bf{r}} = \bA\by={\left( {{{\bf{G}}^H}{\bf{G}}} \right)^{ - 1}}{\bf{G}}^H\by =\sqrt P {\bf{x}} + {\left( {{{\bf{G}}^H}{\bf{G}}} \right)^{ -1}}{\bf{G}}^H{\bf{n}}.
\end{align}
 Focusing on the $k$th element of ${\bf{r}}$, we have, from (\ref{detector 3}),
\begin{align}\label{detector 4}
{r_k}=\sqrt P x_k + {\bf{a}}_k^H{\bf{n}},
\end{align}where ${\bf{a}}_k$ is the $k$th column of the matrix ${\bf{A}}$. Since $\bn\sim\mathcal{CN}({\bf 0}_{M \times 1},\bI_M)$, the equivalent noise  ${\bf{a}}_k^H{\bf{n}}$ is a CSCG random variable with zero mean and variance ${\left\| {{{\bf{a}}_k}} \right\|_F^2}$, i.e., ${\bf{a}}_k^H{\bf{n}}\sim\mathcal{CN}(0,{\left\| {{{\bf{a}}_k}} \right\|_F^2})$. There is no interference term in (\ref{detector 4}) due to the ZF matrix structure. It is noteworthy that ZF detector is chosen here for the simplicity of the presentation and our work can be straightforwardly extended to MMSE linear detector.

\subsection{Asymptotic Achievable Rate Analysis}

In this subsection, we analyze the asymptotic achievable rate of the multi-user distributed massive MIMO uplink when $M\rightarrow \infty$, for a general BS antenna deployment and user location.

To assist the analysis, we first prove the following results for very long random
vectors.

\begin{Lemma}\label{Lemma_law_large_number}
Let ${\bf{p}} \buildrel \Delta \over = {\left[
{\begin{array}{*{20}{c}} {{p_1}}&{{p_2}}& \ldots &{{p_M}}
\end{array}} \right]^T}$ and ${\bf{q}} \buildrel \Delta \over = {\left[
{\begin{array}{*{20}{c}} {{q_1}}&{{q_2}}& \ldots &{{q_M}}
\end{array}} \right]^T}$ be independent $M \times 1$ vectors whose elements are i.n.i.d.~zero-mean random variables. Assume that $\Exp\left\{ {{{\left| {{p_i}} \right|}^2}} \right\} = \sigma _{p,i}^2$, $\Exp\left\{ {{{\left| {{p_i}}
\right|}^4}} \right\} < \infty$ and $\Exp\left\{ {{{\left| {{q_i}}
\right|}^2}} \right\} = \sigma _{q,i}^2$, $\Exp\left\{ {{{\left|
{{q_i}} \right|}^4}} \right\} < \infty$ for $i=1,2,\ldots,n$. We have, when $M\rightarrow \infty$,
\begin{align}\label{large number 1}
\frac{1}{M}{{\bf{p}}^H}{\bf{p}}\mathop  \to \limits^{a.s.}
{\frac{1}{M}\sum\limits_{i = 1}^M {\sigma _{p,i}^2} },
\end{align}
and
\begin{align}\label{large number 2}
\frac{1}{M}{{\bf{p}}^H}{\bf{q}}\mathop  \to \limits^{a.s.} 0,
\end{align}where $\mathop  \to \limits^{a.s.}$ denotes the almost sure convergence.
\label{lemma-1}
\end{Lemma}

\begin{proof}
See Appendix \ref{app-A}.
\end{proof}

The results in Lemma \ref{lemma-1} can be seen as generalizations of the
results in Eqs.~(4) and (5) in \cite{Larsson_uplink}. More specifically, the results in \cite{Larsson_uplink} are for very
long random vectors with i.i.d.~elements, while the results in Lemma
\ref{lemma-1} can be applied to very long random vectors with i.n.i.d.~elements.

With the results in Lemma \ref{lemma-1}, we can obtain the following
expression for the asymptotic achievable rate of distributed massive
MIMO systems.
\begin{Proposition}\label{capacity total}
When $M\rightarrow \infty$, the average achievable rate of the $k$th
user in the multi-user distributed MIMO system has the following
asymptotic behavior:
\begin{align}\label{capacity theorem}
{R_k} \mathop  \to \limits^{a.s.}  {\log _2}\left( {1 +
P\sum\limits_{m = 1}^M {{\beta _{mk}}} } \right).
\end{align}
\end{Proposition}
\begin{proof}
From (\ref{detector 4}), the ergodic achievable rate of the $k$th user is
\begin{equation}\label{capacity 4}
{R_k} = \Exp\left\{ {{{\log }_2}\left( {1 +
\frac{P}{\left\|{{\bf{a}}_k} \right\|^2}} \right)} \right\}=\Exp\left\{ {{{\log }_2}\left( {1 + \frac{P}{{{{\left[
{{{\left( {{{\bf{G}}^H}{\bf{G}}} \right)}^{ - 1}}} \right]}_{kk}}}}}
\right)} \right\}.
\end{equation}
From Lemma \ref{lemma-1}, when $M\rightarrow \infty$, we have
\begin{align}\label{proof2 2}
\frac{1}{M}{\left\| {{{\bf{g}}_k}} \right\|^2} & \mathop  \to
\limits^{a.s.} \frac{1}{M}\sum\limits_{i = 1}^M {\Exp\left\{ {{{\left|
{{g_{ik}}} \right|}^2}} \right\}}  = \frac{1}{M}\sum\limits_{i =
1}^M {{\beta _{ik}}}, \\
\label{proof2 3} \frac{1}{M}{\bf{g}}_k^H{{\bf{g}}_i} & \mathop  \to
\limits^{a.s.} 0,i \ne k.
\end{align}

Entries of ${\left( {{{\bf{G}}^H}{\bf{G}}} \right)}^{ - 1}$ are
continuous and have finite first order derivatives with respect to
${\bf{g}}_1,{\bf{g}}_2,\ldots,{\bf{g}}_K$ when $\bG$ is
non-singular. Meanwhile, the probability of $\bG$ is singular is
zero. Thus, we have
\def\diag{{\rm diag}}
\begin{eqnarray}
&&{\left( {{{\bf{G}}^H}{\bf{G}}} \right)^{ - 1}}
=\frac{1}{M}{\left( \frac{1}{M}{{{\bf{G}}^H}{\bf{G}}} \right)^{
- 1}}
\nonumber \\
&& \overset{a.s}{\to} \frac{1}{M}\diag\left\{\frac{\sum_{i = 1}^M
\beta_{i1}}{M}, \frac{\sum_{i = 1}^M \beta _{i2}}{M}, \frac{\sum_{i
= 1}^M \beta _{iK}}{M} \right\}^{ - 1} \nonumber\\
&&=\diag\left\{\sum_{i = 1}^M
\beta_{i1}, \sum_{i = 1}^M \beta _{i2}, \sum_{i = 1}^M \beta _{iK}
\right\}^{ - 1} \label{proof2 0 5}
\end{eqnarray}
By substituting (\ref{proof2 0 5}) into (\ref{capacity 4}), the proposition is proved.
\end{proof}

\def\bI{{\bf I}}
The result in Eq. (\ref{capacity theorem}) of Proposition
\ref{capacity total} can be applied to massive MIMO systems with
arbitrary antenna deployment and user location. Proposition
\ref{capacity total} is also applicable to MMSE detector, where
$\bA=\left(\bG^H\bG+\frac{1}{P}\bI_k\right)^{-1} \bG^{H}$. With the
MMSE detector, we can change Eqs.~(\ref{detector 3}), (\ref{detector
4}), and (\ref{capacity 4}) accordingly to include the interference.
But due to the large number of antennas, the interference term will
diminish and the same achievable rate result can be obtained.

The multi-user centralized MIMO system considered in
\cite{Marzetta,Massive_5,Massive_6,Larsson_uplink} can be seen as a
special case of our multi-user distributed MIMO system with all the
distributed antennas located in the same place, i.e.,
$\beta_{mk}=\beta_k$ for the $k$th user. Thus, from (\ref{capacity
theorem}), we can obtain its achievable rate result as
\begin{align}\label{capacity case 1}
{R_{k,{\rm central}}} \mathop  \to \limits^{a.s.}  {\log _2}\left(
{1 + PM \beta_k} \right),
\end{align}which is the same as Eq.~(13) in \cite{Larsson_uplink}.

\section{Asymptotic Analysis for the Achievable Rate of Circularly Distributed Antennas}
Theoretically speaking, antennas in a distributed massive MIMO
system can take arbitrary locations and topology. The optimization
of the antenna locations can be highly challenging, if not
intractable, due to the large number of antennas and design
parameters. On the other hand, arbitrary antenna locations or
optimal topology may have prohibitive backhaul cost and installation
cost. In real applications, it is more practical to consider
manageable antenna topology. In this work, we consider circularly
located BS antennas, where all antennas are on a circle centered at
the cell enter. Circular topology has ideal symmetry and low
dimension (radius and angle). Compared with the line topology, it is
expected to have superior performance due to better symmetry. Compared
with the grid topology, it is expected to have lower implementation
cost and more tractability in analysis. Circular antenna layout has
been considered in the literature
\cite{Antenna_placement,D_massive_circle_1,D_massive_circle_2} and
shown to have good performance.

\begin{figure}[!t]
\begin{center}
\includegraphics[width=5in]{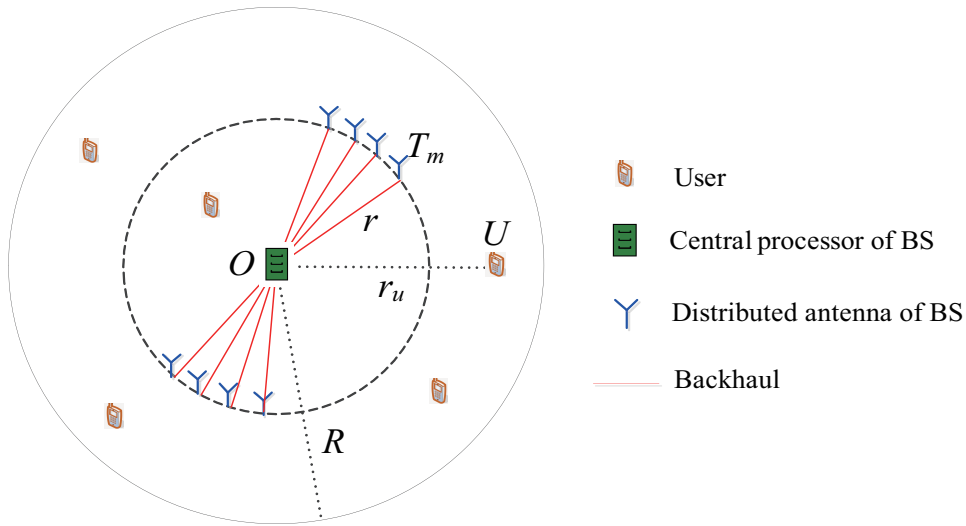}
\caption{Distributed massive MIMO system with circular BS antennas.}
\label{circle-eps}
\end{center}
\end{figure}
In this section, for distributed massive MIMO systems with circular antenna layout, we first specify the system model, then analyze the asymptotic achievable rate for an arbitrary user, and finally derive the average sum-rate per user assuming uniform user location.

The single-cell multi-user distributed massive MIMO system is shown
in Figure \ref{circle-eps}. We consider a circular cell with radius
$R$. The center of the cell is denoted as $O$. Circular cell is
widely used
\cite{Antenna_placement,D_massive_circle_1,D_massive_circle_3} and
has been shown to have similar performance to hexagonal cell. But it
enjoys more tractable analysis. The distributed BS antennas are
located evenly on a smaller circle with radius $r$, whose center is
the same as the center of the cell. We denote the location of the
$m$th BS antenna as $T_m$. Thus, the length of the segment $OT_m$ is
$r$. Notice that since the BS antennas are evenly located on the
circle and the antenna number is assumed to be large, the angle
dimension of the antenna location (for example, the angle of the
segment $OT_m$ and the horizontal axis) has little effect on the
system performance and only the radius of the antenna circle
matters. We denote the location of an arbitrarily user as $U$. Let
$r_u$ be the distance between the user and the cell center, i.e.,
the length of the segment $OU$.

\subsection{Asymptotic Achievable Rate of an Arbitrary User for Circularly Distributed Massive MIMO}

The following results on the asymptotic achievable rate of an arbitrary user at distance $r_u$ for the circularly distributed massive MIMO system are obtained.
\begin{theorem}\label{capacity circle theorem label} For the circularly distributed massive MIMO system with antenna radius $r$, when $M\rightarrow\infty$, the ergodic achievable rate of the user at distance $r_u$ from the cell center has the following asymptotic behavior:
\begin{align}\label{capacity circle theorem}
{R} \mathop  \to \limits^{a.s.}  R_{asy}\triangleq {\log _2}\left( {1 + PM{I_0}}
\right),
\end{align}where
\begin{align}\label{capacity circle theorem 1}
{I_0} \triangleq  {\left| {{r^2} - {r_u^2}} \right|^{ -
\frac{v}{2}}}{P_{\frac{v}{2} - 1}}\left( {\frac{{{r^2} +
{r_u^2}}}{{\left| {{r^2} - {r_u^2}} \right|}}} \right),
\end{align}
with $P_{\cdot}(\cdot)$ the Legendre function \cite{rvq:Table_of_Integrals}.

For $\nu=2,4,6$, closed-form expressions for the achievable rate can obtained as follows:
\begin{align}
{R}\mathop  \to \limits^{a.s.}R_{asy}\triangleq \left\{ {\begin{array}{ll}
{{{\log }_2}\left( {1 + {PM}\frac{1}{{{{\left| {{r^2} - r_u^2} \right|}^{}}}}} \right),}&{\mbox{if $v = 2$},}\\
{{{\log }_2}\left( {1 + PM\frac{{{r^2} + r_u^2}}{{{{\left| {r^2 - r_u^2} \right|}^3}}}} \right),}&{\mbox{if $v = 4$},}\\
{{{\log }_2}\left[ {1 + PM\cdot\frac{3\left(\frac{r^2 + r_u^2}{r^2 - r_u^2}\right)^2-1} {2\left| {r^2 - r_u^2} \right|^3}}
\right],}&{\mbox{if $v = 6$}.}
\end{array}} \right.
\label{cl-form-R}
\end{align}

\end{theorem}

\begin{proof}
See Appendix \ref{app-B}.
\end{proof}

In this theorem, to facilitate the presentation, we introduce a new
notation, $R_{asy}$ for the asymptotic ergodic achievable rate of a
user when $M\rightarrow \infty$. Note that ${P_a}\left( b \right) =
F\left( { - a,a + 1;1;\frac{{1 - b}}{2}} \right)$, where
$F(\cdot,\cdot;\cdot;\cdot)$ is the Gauss hypergeometric function
\cite{rvq:Table_of_Integrals}. Many software for scientific
computations such as Matlab have this function. Thus the result in
(\ref{capacity circle theorem}) and (\ref{capacity circle theorem
1}) of Theorem \ref{capacity circle theorem label} can be easily
calculated numerically. However, due to the special function, for
$\nu\ne 2,4,6$, the achievable rate is not in closed-form and little
insight can be obtained on the performance behavior of the
circularly distributed massive MIMO system with respect to the cell
size and antenna location. Thus, in what follows, we derive bounds
on the achievable rate in closed-form.

\begin{theorem} \label{capacity circle bounds theorem label}Define
\begin{align}\label{capacity circle upper theorem}
R^{B,1} &\triangleq {\log _2}\left( {1 + PM\frac{{{{\left( {{r_u^2}
+ {r^2}} \right)}^{\frac{v}{2} - 1}}}}{{{{\left| {{r_u^2} - {r^2}}
\right|}^{v - 1}}}}} \right),\\
 \label{capacity circle upper theorem_1} R^{B,2} &\triangleq {\log _2}\left( {1 + PM
\frac{{{{\left( {r_u^2 + {r^2}} \right)}^{\frac{v}{2} - 1}}}}{\left|
{r_u^2 - {r^2}} \right|^{v - 1}}\cdot \frac{{{2^{3\left(\frac{v}{2}
- 1\right)}}\Gamma^2 \left( {\frac{v}{2} - \frac{1}{2}}
\right)}}{{\pi \Gamma \left( {v - 1} \right)}} } \right).
\end{align}
For the circularly distributed massive MIMO system with antenna radius $r$, the asymptotic achievable rate of a user at distance $r_u$ from the cell center, denoted as $R_{asy}$, can be bounded as follows
\begin{equation}
\left\{\begin{array}{ll}
R^{B,1}\ge R_{asy}\ge R^{B,2}, & \mbox{ if \ $ 2\le \nu \le 4$}, \\
R^{B,1}\le R_{asy}\le R^{B,2}, & \mbox{ if \ $6 \ge \nu \ge 4$}, \\
R_{asy}=R^{B,1}= R^{B,2}, & \mbox{ if \ $ \nu =2 \ or \ 4$}. \\
\end{array}\right.
\label{bound-R}
\end{equation}

\end{theorem}
\begin{proof}
See Appendix \ref{app-C}.
\end{proof}

Theorem \ref{capacity circle bounds theorem label} provides both lower and upper bounds on the achievable rate. When $v \ge 4$, $R^{B,1}$ is a lower bound and $R^{B,2}$ is an
upper bound; when $v \le 4$, $R^{B,1}$ is an upper bound and
$R^{B,2}$ is a lower bound. When $\mu=2,4$, $R^{B,1}$ and
$R^{B,2}$ are the same and equal the achievable rate of the user. Moreover, it is evident that the bounds in Theorem \ref{capacity circle bounds theorem
label} are in closed-form. We note that the coefficient ${\frac{{{2^{3\left(\frac{1}{2}v - 1\right)}}\Gamma^2 \left( {\frac{v}{2} -
\frac{1}{2}} \right)}}{{\pi \Gamma \left( {v - 1} \right)}}}$ in (\ref{capacity circle upper theorem_1}) only
depends on $\nu$, the path-loss exponent and can be easily calculated offline.

To justify the tightness of the two closed-form bounds, we
analyze their difference as
follows:
\begin{eqnarray*}
\label{capacity bounds differences} \left| {R^{B,1} - R^{B,2}}
\right| &=& \left|\log_2\frac{1+PM\frac{\left(r_u^2 +r^2
\right)^{\frac{v}{2} - 1}} {\left|r_u^2 - r^2\right|^{v -
1}}}{1+PM\frac{\left(r_u^2 +r^2 \right)^{\frac{v}{2} -
1}}{\left|r_u^2 - r^2\right|^{v - 1}} \cdot
\frac{2^{3\left(\frac{v}{2}
- 1\right)}\Gamma^2 \left(\frac{v}{2} - \frac{1}{2} \right)}{\pi \Gamma \left( {v - 1} \right)}}\right| \\
& \le &\left| {{{\log }_2}{\frac{{{2^{3\left(\frac{v}{2} -
1\right)}}\Gamma^2 \left( {\frac{v}{2} - \frac{1}{2}} \right)}}{{\pi
\Gamma \left( {v - 1} \right)}}}} \right| \mathop  <
\limits^{(a)}0.6,
\end{eqnarray*}
where (a) is obtained by software calculations for $2 \le v \le 6$.
This shows that the two bounds are close to each other with less
than 0.6 bits/s/Hz difference. The difference is negligible for
massive MIMO systems as both bounds $R^{B,1}$ and $R^{B,2}$ increase
in $\log_2 M$. Thus, either $R^{B,1}$ and $R^{B,2}$ can function as
a tight closed-form approximation of the achievable rate when the
number of BS antennas is large. Our simulation results in Section
\ref{sec-simu} also justify the tightness of the bounds.

Since both bounds $R^{B,1}$ and $R^{B,2}$ increase in $\log P$ where $P$ is the transmit power and $\log M$ where $M$ is the number of BS antennas, the user
achievable rate is proved to increase in $\log P$ and $\log M$.

\subsection{Asymptotic Average Achievable Rate of the Cell for Circularly Distributed Massive MIMO}
In the previous subsection, we have analyzed the asymptotic rate of
an arbitrarily located user in the cell. In this subsection, we
derive the asymptotic average rate of a user in the cell, which
indicates the average experience of user service. The users are
assumed to be randomly and uniformly located
\cite{Antenna_placement,user_location}. Thus, the probability
distribution function of a user's distance to the cell center,
denoted as $r_u$, is
\begin{align}\label{circle average 0}
f_{r_u}\left( x \right) = \frac{{2}}{{{R^2}}}x.
\end{align}
The angle of the user's location (to the horizonal axis) is uniform distributed on $[0,2\pi)$. The following theorem on the asymptotic average rate of a user is proved.

\begin{theorem}\label{circle average capacity} Define
\begin{align}\label{circle average 3}
{{\bar R}^{B,1}}&= {\log _2}\left( {PM} \right) + \left(
{\frac{v}{2} - 1} \right)\left( {1 + \frac{{{r^2}}}{{{R^2}}}}
\right){\log _2}\left( {{R^2} + {r^2}} \right) - \left( {v - 1}
\right)\left( {1 - \frac{{{r^2}}}{{{R^2}}}} \right){\log _2}\left(
{{R^2} - {r^2}}
\right)\notag\\
&\hspace{1cm}- \left( {3v-4} \right)\frac{{{r^2}}}{{{R^2}}}{\log
_2}r + \frac{v}{2}{\log _2}e
,\\
\label{circle average 3 1} {{\bar R}^{B,2}}&= {{\bar R}^{B,1}} +
{\log _2} {\frac{{\Gamma^2 \left(
{\frac{v}{2} - \frac{1}{2}} \right)}}{{\pi \Gamma \left( {v - 1} \right)}}}
+3\left(\frac{v}{2} - 1\right).
\end{align}
For the circularly distributed massive MIMO system with uniformly distributed users in the cell, the asymptotic average achievable rate per user of the cell, denoted as ${\bar R}_{asy}$, can be
bounded as follows:
\begin{equation}
\left\{\begin{array}{ll}
{\bar R}^{B,1}\gtrsim {\bar R}_{asy}\gtrsim {\bar R}^{B,2}, & \mbox{ if \ $ 2\le \nu \le 4$}, \\
{\bar R}^{B,1}\lesssim {\bar R}_{asy}\lesssim {\bar R}^{B,2}, & \mbox{ if \ $6 \ge \nu \ge 4$}, \\
{\bar R}_{asy}\approx{\bar R}^{B,1}\approx {\bar R}^{B,2}, & \mbox{ if \ $ \nu =2 \ or \ 4$}. \\
\end{array}\right.
\label{asy-R}
\end{equation}

\end{theorem}

\begin{proof}
With uniformly distributed user location and the probability density  function of the user distance to the cell center in (\ref{circle average 0}), the asymptotic average achievable rate per user of the cell can be calculated as:
\begin{align}\label{circle average 2}
{\bar R}_{asy}=\frac{2}{{{R^2}}}\int_0^R r_u R_{asy}(r_u) d (r_u),
\end{align}
where $R_{asy}(r_u)$ is the asymptotic achievable rate for a user at distance $r_u$. From (\ref{bound-R}) in Theorem \ref{capacity circle bounds theorem label}, we have (\ref{asy-R}), where
\[
{\bar R}^{B,1}\triangleq \frac{2}{{{R^2}}}\int_0^R r_u R^{B,1}(r_u) d r_u,\quad
{\bar R}^{B,2}\triangleq \frac{2}{{{R^2}}}\int_0^R r_u R^{B,2}(r_u) d r_u.
\]
From (\ref{capacity circle upper theorem}) in Theorem \ref{capacity circle bounds theorem label}, we have
\setlength{\arraycolsep}{1pt}
\begin{eqnarray}
{\bar R}^{B,1} &=&\frac{2}{{{R^2}}}\int_0^R {r_u{{\log }_2}\left( {1 +
PM\frac{{{{\left( {{r_u^2} + {r^2}} \right)}^{\frac{v}{2} -
1}}}}{{{{\left| {r_u^2 - r^2} \right|}^{v - 1}}}}} \right)dr_u}\label{pre-approx}\\
&\approx&\frac{2}{{{R^2}}}\int_0^R {r_u{{\log }_2}\left( {PM\frac{{{{\left( {{r_u^2} + {r^2}} \right)}^{\frac{v}{2} -
1}}}}{{{{\left| {{r_u^2} - {r^2}} \right|}^{v - 1}}}}} \right)dr_u}\label{appro-R}\\
&=& \frac{\log_2e}{R^2}\left[\int_0^{R^2} \ln\left( {PM} \right)dt
+ \left({\frac{v}{2} - 1} \right)\int_0^{R^2}\ln\left(t+r^2\right)dt \right.\\
&&\hspace{6cm} \left. -\left(v - 1\right)\int_0^{R^2}\ln\left| t - r^2\right| dt\right]\nonumber \\
&=&\log_2(PM)+\left(\frac{v}{2}- 1\right)\log_2e
\left[\ln(R^2+r^2)+\frac{r^2}{R^2}\ln\frac{R^2+r^2}{r^2}-1\right] \nonumber\\
&&\hspace{2.5cm} - \frac{v - 1}{R^2}\log_2e\left[\int_0^{r^2} \ln(
r^2-t)dt
+\int_{r^2}^{R^2} \ln(t-r^2) dt\right] \nonumber\\
&=&\log_2(PM)+\left(\frac{v}{2}- 1\right)\log_2e
\left(\ln(R^2+r^2)+\frac{r^2}{R^2}\ln\frac{R^2+r^2}{r^2}-1\right) \nonumber\\
&&\hspace{2cm} - (v - 1)\log_2e\left[\frac{r^2}{R^2}\ln r^2
+\left(1-\frac{r^2}{R^2}\right)\ln(R^2-r^2)-1\right], \nonumber
\end{eqnarray}
\setlength{\arraycolsep}{5pt}
from which we can obtain (\ref{circle average 3}) via simple rewriting.

In deriving (\ref{appro-R}), we have used the approximation $\log(1+x)\approx\log x$ for $x\gg1$. When $M\gg1$, we can see from (\ref{pre-approx}) that the approximation applies. With straightforward and similar calculations, we can obtain (\ref{circle average 3 1}).
\end{proof}

It is evident that our derived bounds on the asymptotic average
achievable rate in Theorem \ref{circle average capacity} are in
closed-form. Also, calculating the difference between the two bounds, we
have\begin{align}\label{capacity bounds differences} \left| {{\bar
R}^{B,1} - {\bar R}^{B,2}} \right| = \left| {\log _2}
{\frac{{\Gamma^2 \left( {\frac{v}{2} - \frac{1}{2}} \right)}}{{\pi
\Gamma \left( {v - 1} \right)}}} +3\left(\frac{v}{2} - 1\right)
\right|   < 0.6.
\end{align}
Thus, either bound can be used as a tight approximation of
$\bar{R}_{asy}$ with the error being less than 0.6 bits/s/Hz. The
error is negligible when $M \gg 1$ since $\bar{R}_{asy}$ increases
in logarithm in $M$. The tightness of the bounds will also be justified by our simulation results in Section \ref{sec-simu}.

\section{Location Optimization of the Circular Antenna Array}
In the previous section, the ergodic achievable rate of an
arbitrarily located user and the average achievable rate per user of
the cell for uniformly located users are derived when $M\rightarrow
\infty$. We can see from the results that other than the transmit
power $P$ and the number of the distributive antennas $M$, the
radius of the distributed antennas $r$ largely affects the average
achievable rate. In this section, we turn to derive the optimal
radius of the circularly distributed antenna array to maximize the
average achievable rate of the cell, which is one of the most
important measures of wireless system performance.

In Theorem \ref{circle average capacity}, both the upper and lower bounds, ${\bar R}^{B,1}$ and ${\bar R}^{B,2}$, are derived for the average achievable rate per user. The bounds are in closed-form and shown to be close to each other. Thus, in the radius optimization, we aim at maximizing ${\bar R}^{B,1}$. The same result can be obtained if ${\bar R}^{B,2}$ is used since the difference ${\bar R}^{B,2}-{\bar R}^{B,1}$ only depends on $v$ and is independent of $r$, the radius of the circular antenna array.

\begin{Lemma}\label{circle average opt theorem}
The radius of the circular antenna array for the distributed massive MIMO system that maximizes ${\bar R}^{B,1}$ is:
\begin{align}\label{circle average opt}
{r_{opt}} = \sqrt {\frac{{{R^2}}}{{{t_0} + 1}}}
\end{align}where $t_0$ is the solution of the following equation:
\begin{align}\label{circle average 6}
x^{3+\frac{{2}}{{v - 2}}} + 2x^{2+\frac{2}{v - 2}} - 1 = 0.
\end{align}
\end{Lemma}

\begin{proof}
The derivative of $\bar{R}^{B,1}$ in (\ref{circle average
3}) with respect to $r$ can be calculated to be:
\begin{equation}\label{circle average 4}
\frac{{d{{\bar R}^{{B_1}}}}}{{dr}} = \frac{{r{{\log }_2}e}}{{{R^2}}}\left\{ {\left( {v - 2} \right)\ln \left( {\frac{{{R^2}}}{{{r^2}}} + 1} \right) + \left( {2v - 2} \right)\ln \left( {\frac{{{R^2}}}{{{r^2}}} - 1} \right)} \right\}
\end{equation}
By making $d{{\bar R}^{{B_1}}}/dr$ zero, we have\begin{equation} {\left( {\frac{{{R^2}}}{{{r^2}}} + 1}
\right)^{\frac{{v/2 - 1}}{{v - 1}}}}  \left(
{\frac{{{R^2}}}{{{r^2}}} - 1} \right) = 1. \label{circle average 5}
\end{equation}
 After replacing $R^2/r^2 - 1$ with $t$ and rearranging the
expression in (\ref{circle average 5}), we obtain (\ref{circle
average 6}).

Next we show that the solution of (\ref{circle average 5}), denoted as $r_{opt}$, is the maximum of ${\bar R}^{{B_1}}$. From (\ref{circle average 4}), we can see that $\left(\frac{{r{{\log }_2}e}}{{{R^2}}}\right)^{-1}\frac{d{{\bar R}^{{B_1}}}}{dr}$ is a decreasing function of $r$. Thus we have $\left(\frac{{r{{\log }_2}e}}{{{R^2}}}\right)^{-1}\frac{d{{\bar R}^{{B_1}}}}{dr}>0$ when $r<r_{opt}$ and $\left(\frac{{r{{\log }_2}e}}{{{R^2}}}\right)^{-1}\frac{d{{\bar R}^{{B_1}}}}{dr}<0$ when $r>r_{opt}$. Notice that $\frac{{r{{\log }_2}e}}{{{R^2}}}>0$. Thus we have $d{{\bar R}^{{B_1}}}/dr>0$ when $r<r_{opt}$ and $d{{\bar R}^{{B_1}}}/dr<0$ when $r>r_{opt}$.  This ends the proof.
\end{proof}

The equation in (\ref{circle average 6}) only depends on the path-loss exponent $v$ and can be easily solved offline by many softwares such as Matlab. Using the result in Lemma \ref{circle average opt theorem}, the radius of the distributed circular BS antenna array can be designed for the maximum average achievable rate for the distributed massive MIMO system.

From (\ref{circle average opt}) we have $r_{opt}/R=1/\sqrt{t_0+1}$.
Thus, the ratio of the antenna radius and the cell radius depends on
the path-loss exponent $v$ only and is independent of the transmit
power $P$ and the BS antennas size $M$. This is very appealing in
wireless network designs and implementation. Based on this fact, we
further note that the improvement of the hardware in distributed
MIMO, such as increasing the number of the distributed antennas,
will not affect the optimal location of the distributive antennas.

\section{Numerical results}
\label{sec-simu}

In this section, we present numerical results to show the
performance of the distributed massive MIMO system with circular BS
antenna array and justify the accuracy of our theoretical results.
The impacts of different parameters, such as the number and the
location of the distributive antennas, the transmit power, and the
path-loss exponent, on the achievable rate are also investigated. We
consider the uplink of a circular cell whose cell radius is set as
$R=1000$ meters. There is a massive BS with $M$ circularly
distributed antennas located on a circle of radius $r$. There are
$K=9$ users. The users have the same transmit power, which is set to
be $P \times r_{mid}^{v}$ where $r_{mid}=R/2=500$ meters. So, if a
user is located 500 meters away from a BS antenna, the average
received SNR of the antenna from the user is $P$. The normalization
with $r_{mid}^{v}$ in the transmit power does not affect the
behavior of the simulation curves but only affects the position of
the curves on the $P$ axis. The small-scale channel fading $h_{mk}$
is generated as circularly symmetric complex Gaussian with zero-mean
and unit-variance, thus Rayleigh fading.

\subsection{Achievable Rate of an Arbitrarily Located User}
In Figure \ref{bound_capacity_new.eps} and Figure
\ref{capacity_different_M_single_user.eps}, we show the simulated
achievable rate of an arbitrarily user and compare with the derived
asymptotic analytical results in (\ref{capacity circle
theorem}-\ref{capacity circle theorem 1}), as well as the
closed-form bounds $R^{B,1}$ in (\ref{capacity circle upper
theorem}) and $R^{B,2}$ in (\ref{capacity circle upper theorem_1}).
We set $r=500$ meters. Figure \ref{bound_capacity_new.eps} shows the
achievable rate for different user location $r_u$ (the distance of
the user to the cell center) and path-loss exponent $v$, while
$M=300$ and $P=10$dB. Figure
\ref{capacity_different_M_single_user.eps} shows the achievable rate
for different antenna number $M$ and user transmit power $P$, while
$r_u=300$ meters and $v=3.6$ \cite{rvq:david}.

We can see from the figures that the curves numerical calculated by
(\ref{capacity circle theorem}-\ref{capacity circle theorem 1})
accurately predict the simulated ones. The derived closed-form
bounds in (\ref{capacity circle upper theorem}) and (\ref{capacity
circle upper theorem_1}) are very close to the simulated and
numerically calculated achievable rates. The figures show that $R^{B,1}$
is a lower bound and $R^{B,2}$ is an upper bound when $v \ge 4$,
while $R^{B,1}$ is an upper bound and $R^{B,2}$ is a lower bound
when $v \le 4$, which confirms the results of (\ref{bound-R}).

\begin{figure}[!t]
\centering
\includegraphics[width=5.5in]{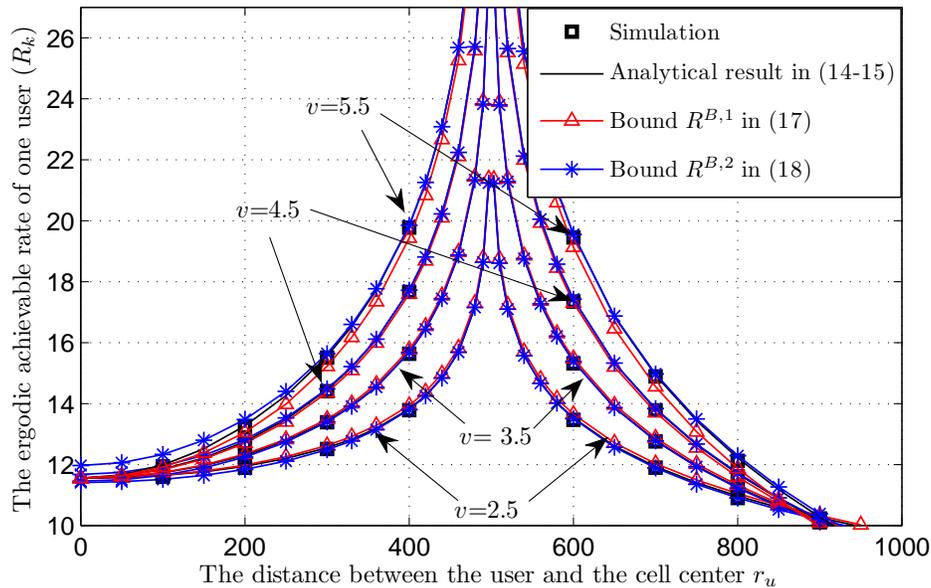}
\caption{Comparison of the analytical expressions and bounds of the
ergodic achievable rate of distributed massive MIMO with simulation,
where $K=9$, $r=500$ meters, $M=300$ and $P=10$dB.}
\label{bound_capacity_new.eps}
\end{figure}

\begin{figure}[!t]
\centering
\includegraphics[width=5.5in]{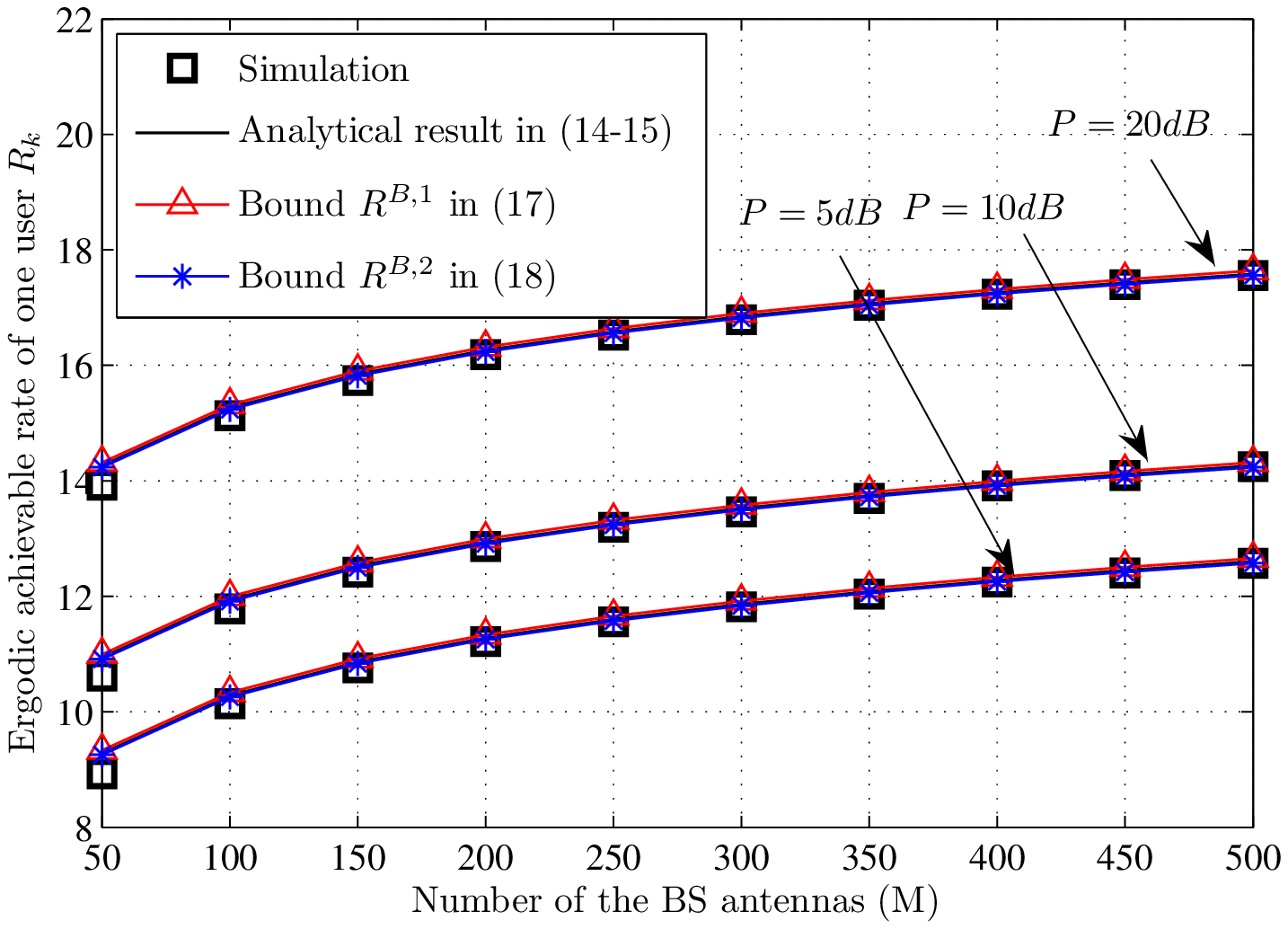}
\caption{The ergodic achievable rate of one user for different $M$
where $K=9$, $r=500$ meters, $r_u=300$ meters, and $v=3.6$.}
\label{capacity_different_M_single_user.eps}
\end{figure}


Figure \ref{bound_capacity_new.eps} also shows that the achievable
rate is higher for larger path-loss exponent and smaller distance
between the user distance $r_u$ and the radius of the circular
antenna array $r$. For either $r_u>r$ or $r_u<r$, the achievable
rate is a concave function of $r_u$. We can see from Figure
\ref{capacity_different_M_single_user.eps} that the achievable rate
increases with $M$. For example, increasing $M$ from $100$ to $400$
brings an achievable rate increase of  about $16\%$ at $P=10$dB. For
the $P=20$dB case, increasing $M$ from $100$ to $400$ brings about
an achievable rate increase of $12\%$. The achievable rate also
increases with $P$, the user transmit power. For example, increasing
$P$ from $5$dB to $20$dB results in the achievable rate increase of
about $40\%$ at $M=300$.

\begin{figure}[!t]
\centering
\includegraphics[width=5.5in]{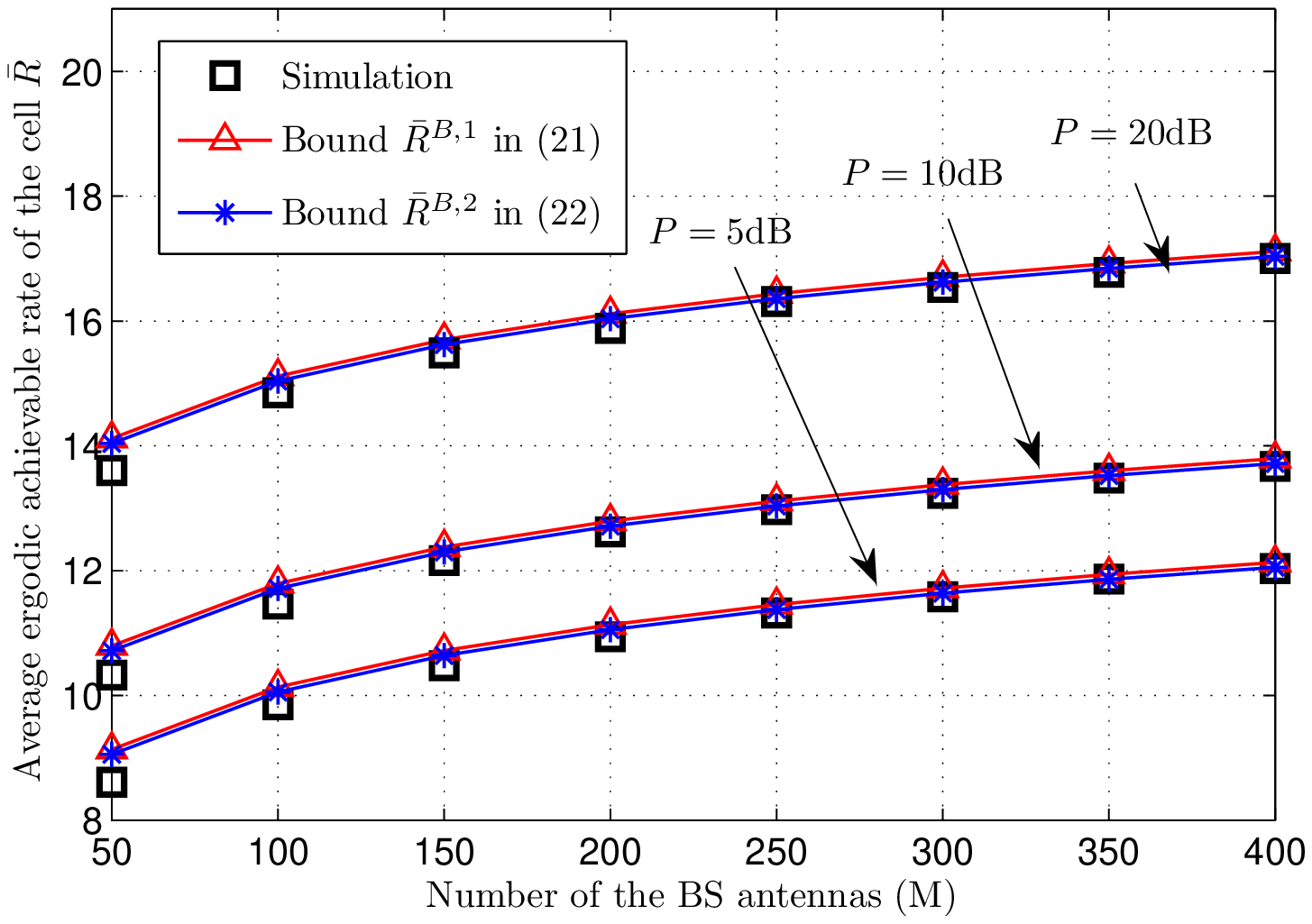}
\caption{The average ergodic achievable rate of the cell for
different $M$ where $K=9$, $r=500$ meters, and $v=3.6$.}
\label{capacity_different_M_the_cell.eps}
\end{figure}

\begin{figure}[!t]
\centering
\includegraphics[width=5.5in]{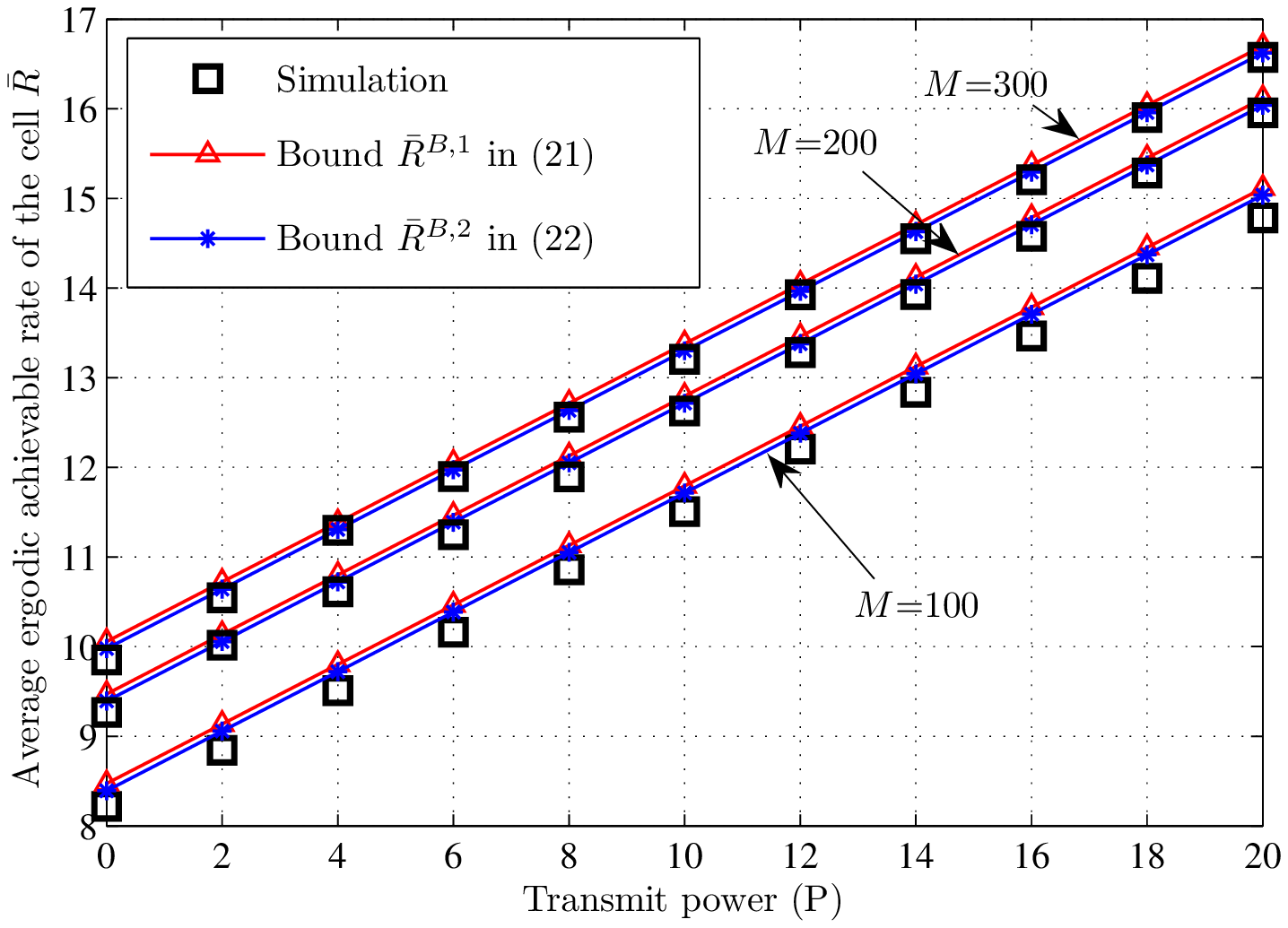}
\caption{The average ergodic achievable rate of the cell for
different $P$ where $K=9$, $r=500$ meters, and $v=3.6$.}
\label{capacity_different_P_the_cell.eps}
\end{figure}
\subsection{Average Achievable Rate  of the Cell}

In Figure \ref{capacity_different_M_the_cell.eps} and Figure
\ref{capacity_different_P_the_cell.eps}, we show the average
achievable rate per user of the cell, and compare with the derived
bounds $\bar{R}^{B,1}$ in (\ref{circle average 3}) and
$\bar{R}^{B,2}$ in (\ref{circle average 3 1}). We set $r=500$ meters
and assume a practical urban scenario with the path-loss exponent
$v=3.6$. The user location are randomly generated to be uniformly
distributed in the cell.

It can be seen from both figures that the simulated achievable rate
and the derived closed-form bounds match well for all adopted values
of $M$ and $P$. The average achievable rate of the cell increases
with $M$, which indicates that increasing the number of the BS
antennas improves the system throughput. For example, the $M=400$
scenario achieves about $15\%$ higher average rate than the $M=100$
scenario at $P=10$dB. The average achievable rate also increases
with $P$. For example, increasing $P$ from $4$dB to $14$dB brings an
achievable rate advantage of about $35\%$ at $M=100$.

\subsection{Impact of the Location of the Circularly Antenna Array}
Next, we show the impart of the radius of circular antenna array,
$r$, on the average achievable rate of the cell. Notice that the
$r=0$ case corresponds to centralized massive MIMO system, where the
BS antennas are located at the center of the cell.

Figure \ref{capacity_r_M_and_P_average.eps} plots the simulated average rate of the cell and the derived bounds as functions of $r$ for three cases: 1) $M=150,P=10$dB, 2) $M=150, P=20$dB, and 3) $M=300,P=20$dB. We set $v=3.6$. The figure shows that the radius of the distributed antenna array has significant influence on the average rate of the cell and proper antenna location results in significant improvement in average rate to the centralized case. For example, increasing $r$ from $0$ to $750$ meters boosts up the average rate by about $30\%$ when $M=150,P=20$dB. The figure also indicates that the optimal $r$ for different $M$ and $P$ remains the same, which is about $750$ meters. This conforms with our result in Lemma \ref{circle average opt theorem} that the optimal $r$ is irrelevant to the values of $M$ and $P$ but only depends on $v$.

\begin{figure}[!t]
\centering
\includegraphics[width=5.5in]{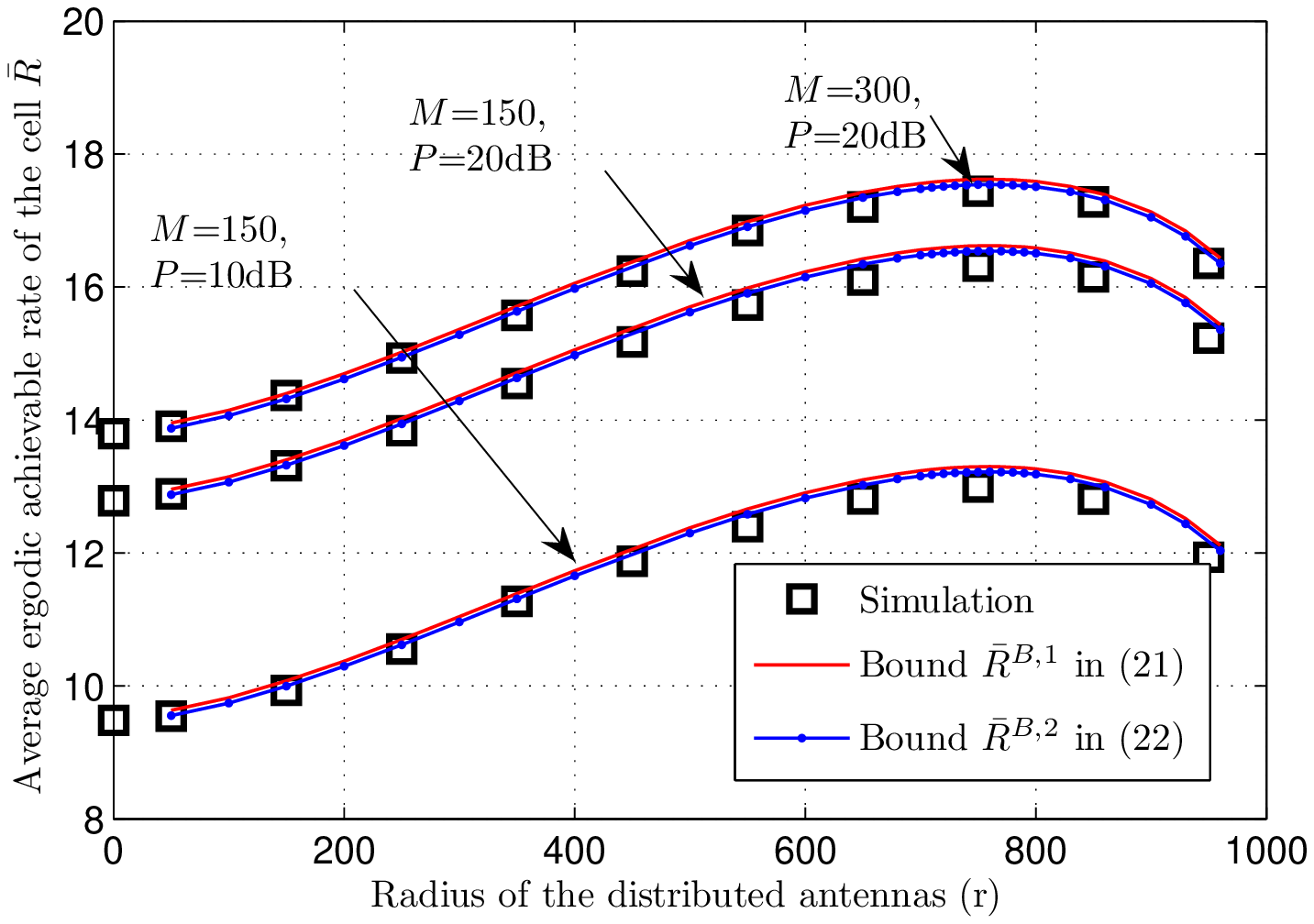}
\caption{The average ergodic achievable rate of the cell for
different $r$ where $M=150,P=10$dB, $M=150, P=20$dB, and $M=300,P=20$dB. We set $K=9$ and $v=3.6$.} \label{capacity_r_M_and_P_average.eps}
\end{figure}

To further understand the optimal radius of the circular antenna
array, employing (\ref{circle average opt}) in Lemma \ref{circle
average opt theorem}, we plot $r_{opt}/R$, the ratio of the optimal
antenna array radius to the cell radius, for different  path-loss
exponent, in Figure \ref{capacity_v_r_optimal.eps}. It can be seen
that $r_{opt}/R$ is bigger for larger $v$. For example, when
$v=3.5$, $r_{opt}/R$ is $0.758$, while it is $0.766$ when $v=4.0$.
Thus, as the path-loss is larger, antennas should be installed
further away from the cell center for the maximum average rate. For
a given $v$ value, the radius of the circular antennas array should
increases linearly in the cell radius $R$. We can also see that for
any $v\in[2,6]$, $r_{opt}/R\in[0.7,0.78]$. This shows that the
optimal $r_{opt}/R$ value is far away from the centralized massive
MIMO case, where $r/R=0$. On the other hand, for different $v$
values within the practical range ($v \in [2,6]$), $r_{opt}/R$ has
small change. Actually, for any $v\in[2,6]$, setting the radius of
the circular antenna array as $r=0.75R$ will induce less than $5\%$
loss in the average rate compared to the optimal radius. This result
is useful in further simplifying the practical system design of
circularly distributed massive MIMO systems.

\begin{figure}[!t]
\centering
\includegraphics[width=5.5in]{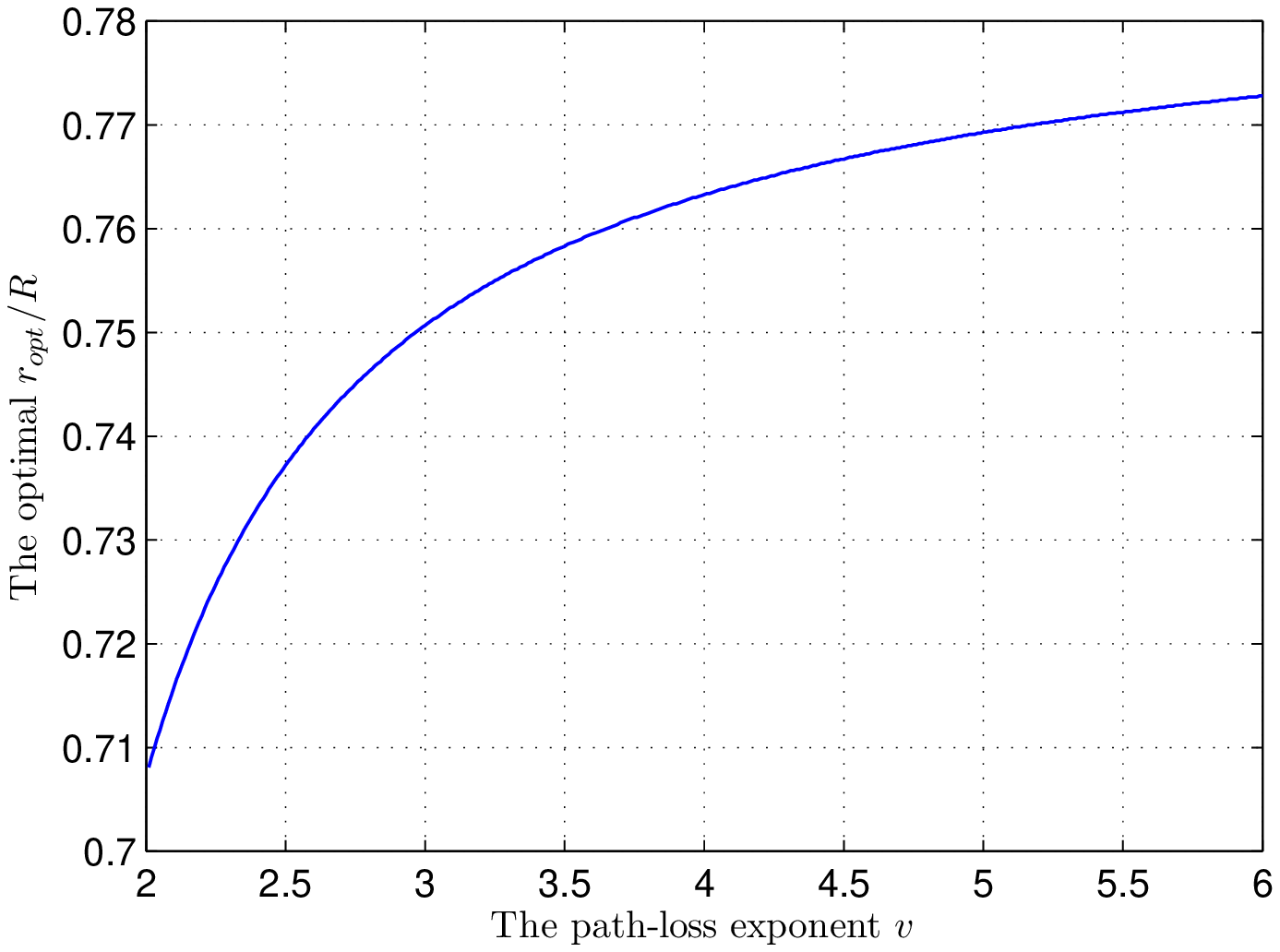}
\caption{The relationship between $r_{opt}/R$ and the
path-loss exponent $v$.} \label{capacity_v_r_optimal.eps}
\end{figure}

\section{Conclusions}
In this paper, we have considered the uplink of a single-cell
multi-user distributed massive MIMO system, where the BS equipped
with a large number of distributed antennas receiving information
from multiple users equipped with single antenna. Zero-forcing
detection is used at the BS. In order to analyze the achievable rate
of the system, we provided new results for very long random vectors
with independent but non-identically distributed entries. Based on
the results, for circularly distributed base station antennas, we
derived analytical expressions of the achievable rate of an
arbitrarily located user and two closed-form expressions that bound
the rate from both sides. The tightness of the bounds were
rigorously justified. From these results, behavior of the system
achievable rate with respect to different parameters such as the
size of the antenna array, the location of the antenna array, the
path-loss exponent, and the transmit power can be understood. We
also derived tight closed-form bounds for the average achievable
rate per user assuming that users are randomly located in the cell,
from which the optimal radius of the distributed antenna array that
maximizes the average rate was derived. Numerical results were
illustrated to justify our analytical results. Our work has shown
that multi-user distributed massive MIMO largely outperforms
centralized massive MIMO. Our derived results can assist
infrastructure providers in solving the fundamental problems of
performance measurement and antennas placement for distributed
massive MIMO systems in practice.

\appendices
\section{Proof of Lemma 1}
\label{app-A}
Let $a_i\triangleq \left| {{p_i}} \right|^2$. Thus $a_i$'s are independent and there exists a finite positive constant $C$ such that
\begin{align}\label{proof 1}
\Exp\left\{ {{a_i}} \right\}
=\Exp\left\{ {{{\left| {{p_i}} \right|}^2}}
\right\} = \sigma _{p,i}^2,\quad
\Exp\left\{ {{a_i^2}} \right\}=\Exp\left\{
{{{\left| {{p_i}} \right|}^4}} \right\} \le C.
\end{align}
The variance of the arithmetic mean of $a_1,a_2,\ldots,a_M$
satisfies the following:
\begin{align}\label{proof 2}
&Var\left\{ {\frac{1}{M}\sum\limits_{i = 1}^M {{a_i}} } \right\} =
\frac{1}{{{M^2}}}\sum\limits_{i = 1}^M {Var\left\{ {{a_i}} \right\}}
\le \frac{C}{M}.
\end{align}

From Chebyshev's inequality, we have
\def\Prob{{\mathbb P}}
\begin{align}\label{proof 3}
&\Prob\left\{ {\left| {\frac{1}{M}\sum\limits_{i = 1}^M {{a_i}}  -
\frac{1}{M}\sum\limits_{i = 1}^M {\Exp\left\{ {{a_i}} \right\}} }
\right| < \varepsilon } \right\}
 \ge 1 - \frac{1}{\epsilon^2}Var\left\{ {\frac{1}{M}\sum\limits_{i = 1}^M {{a_i}} } \right\} \ge 1 - \frac{C}{{M{\varepsilon ^2}}},
\end{align}where (\ref{proof 2}) is used in the last step.

From the definition of $a_i$, we have
\begin{align}\label{proof 4}
 \frac{1}{M}{{\bf{p}}^H}{\bf{p}} = \frac{1}{M}\sum\limits_{i = 1}^M
{ {{{\left| {{p_i}} \right|}^2}} }  = \frac{1}{M}\sum\limits_{i =
1}^M {{a_i}}.
\end{align}
Using this in (\ref{proof 3}), we obtain
\begin{align}\label{proof 5}
1 \ge \Prob\left\{ {\left| {\frac{1}{M}{{\bf{p}}^H}{\bf{p}} -
\frac{1}{M}\sum\limits_{i = 1}^M {\sigma _{p,i}^2} } \right| <
\varepsilon } \right\} \ge 1 - \frac{C}{{M{\varepsilon ^2}}}.
\end{align}When $M \to \infty $, $1 - \frac{C}{{M{\varepsilon ^2}}}\rightarrow 1$. Thus, Eq. (\ref{large number 1}) is
proved.

Since $p_i$ and $q_i$ are independent, we have $\Exp\left\{
{p_i^H{q_i}} \right\} = 0$, $i=1,2,\ldots,M$. Following the same
arguments in the proof of (\ref{large number 1}), Eq. (\ref{large number
2}) can be proved.


\section{Proof of Theorem 3}
\label{app-B}
As shown in Figure \ref{circle-eps}, we use $O$ for the cell center. To help the derivation, we denote the angle of the segments $OT_m$ (where $T_m$ is the location of the $m$th BS antenna) and $OU$ (where $U$ is the location of the user) as $\alpha _{m}$. The distance between $T_M$ and $U$, denoted as ${{D}_{m}} $, can be expressed as
\begin{align}\label{circle case 1}
{{D}_{m}} = \sqrt {{r^2}{{\sin }^2}{\alpha _{m}} + {{\left( {r\cos
{\alpha _{m}} - r_u} \right)}^2}}.
\end{align}
Without loss of generality, we assume that $\alpha _{1}=0$ and
the evenly circularly distributed BS antennas are labeled such that
$\alpha _{m}=\frac{m-1}{M}2\pi$ for $1\le m\le \lfloor\frac{M}{2}\rfloor$ and $\alpha _{m}=\left(\frac{m-1}{M}-1\right)2\pi$ for $ \lceil\frac{M}{2}\rceil\le m\le M$. Let $\Delta \alpha  \triangleq \frac{{2\pi }}{M}$.



From Eqs. (\ref{large scale}) and (\ref{circle case 1}),
\setlength{\arraycolsep}{1pt}
\begin{eqnarray}\label{circle case 2}
\frac{1}{M}\sum\limits_{m = 1}^M {{\beta _{mk}}}
&=& \frac{1}{M}\sum\limits_{m = 1}^M {\frac{1}{{{{D}}_{m}^v}}} \notag \\
 &=& \frac{1}{M}\sum\limits_{m = 1}^M {{{\left[ {{r^2}{{\sin }^2}{\alpha _{m}} + {{\left( {r\cos {\alpha _{m}} - r_u} \right)}^2}} \right]}^{ -
 \frac{v}{2}}}}\\
&=& \frac{1}{M}\frac{1}{{\Delta \alpha }}\sum\limits_{m  = 1}^{M }
{{{\left[ {{r^2}{{\sin }^2}{\alpha _{m}} + {{\left( {r\cos {\alpha _{m}} - r_u} \right)}^2}} \right]}^{ - \frac{v}{2}}}} \Delta \alpha \label{sum}\\
&\overset{M\rightarrow\infty}{\longrightarrow}&
\frac{1}{2\pi}
\int_{-\pi}^{\pi } \left[ r^2\sin ^2\alpha + \left( r\cos\alpha - r_u \right)^2 \right]^{- \frac{v}{2}}d\alpha.\label{asy-int}
\end{eqnarray}
\setlength{\arraycolsep}{5pt}
Employing \cite[2.5.16.38]{integral_and_series_Vol_1}, we have, from (\ref{asy-int}),
\begin{align}\label{circle case 7}
\frac{1}{M}\sum\limits_{m = 1}^M {{\beta _{mk}}} \overset{M\rightarrow\infty}{\longrightarrow} {\left| {{r^2} - {r_u^2}} \right|^{ -
\frac{v}{2}}}{P_{\frac{v}{2} - 1}}\left( {\frac{{{r^2} +
{r_u^2}}}{{\left| {{r^2} - {r_u^2}} \right|}}} \right).
\end{align}
By using (\ref{circle case 7}) in (\ref{capacity theorem}), Eq.~(\ref{capacity circle theorem}) can be proved. Eq.~(\ref{cl-form-R}) can  be subsequently obtained by using Eq.~\cite[8.912]{rvq:Table_of_Integrals}.

To further illuminate the asymptotic result in (\ref{asy-int}), we investigate the difference between (\ref{asy-int}) and (\ref{sum}). To help the presentation, we use the following notation:
\begin{equation*}
f(\alpha)\triangleq\left[ r^2\sin^2\alpha + \left( r\cos\alpha-r_u \right)^2 \right]^{- \frac{v}{2}}.
\end{equation*}
For $\alpha\in[0,\pi]$, we can obtain via straightforward calculations that
\begin{align}\label{circle case diff 1}
\frac{\partial f}{\partial \alpha} =  -v\left[ r^2\sin^2\alpha +
\left( r\cos \alpha - r_u\right)^2\right]^{-\frac{v}{2}-1}
rr_u\sin\alpha \le 0.
\end{align}
This shows that  $f$ decreases with $\alpha$ when $\alpha\in[0,\pi]$. For the simplicity of presentation, we assume that $M$ is even. The proof for odd $M$ is similar. The difference between (\ref{asy-int}) and (\ref{sum}) can be bounded as follows:
\setlength{\arraycolsep}{1pt}
\begin{eqnarray}
{I_{diff}} &\triangleq&  \frac{1}{M\Delta\alpha}
\sum_{m = 1}^M \left[ r^2\sin^2\alpha_{m} + \left( r\cos\alpha_{m}-r_u \right)^2 \right]^{- \frac{v}{2}} \Delta\alpha - \frac{1}{{2\pi }}\int_{ - \pi }^\pi
{f\left( {\alpha} \right)d\alpha } \\
&=& 2 \frac{1}{M\Delta\alpha}
\sum_{m = 1}^\frac{M}{2} \left[ r^2\sin^2\alpha_{m} + \left( r\cos\alpha_{m}-r_u \right)^2 \right]^{- \frac{v}{2}} \Delta\alpha
- \frac{1}{{2\pi }}\int_{0 }^\pi
{f\left( {\alpha} \right)d\alpha } \label{45} \\
&=& \frac{1}{\pi}  \sum_{m = 1}^{\frac{M}{2}}
f_k\left(\frac{m-1}{M}2\pi\right) \frac{2\pi}{M}
-\sum_{m = 1}^{\frac{M}{2}}\int_{\frac{m-1}{M}}^{\frac{m}{M}}
f\left(\alpha\right)d\alpha \\
&\le& \frac{1}{\pi} \sum_{m = 1}^{\frac{M}{2}}
\left[f\left(\frac{m-1}{M}2\pi\right)-
f\left(\frac{m}{M}2\pi\right)\right]\frac{2\pi}{M}\\
&=&\frac{2}{M}[f(0)-f(\pi)] = \frac{2}{M}[(r-r_u)^{-v}-(r+r_u)^{-v}]
\overset{M\rightarrow\infty}{\longrightarrow}0.
\end{eqnarray}
\setlength{\arraycolsep}{5pt}
In obtaining (\ref{45}), we use the symmetry in $f(\alpha)$. This analysis shows that the difference between (\ref{asy-int}) and (\ref{sum}) is linear in $1/M$. For large but finite number of antennas, (\ref{asy-int}) is a tight approximation of (\ref{sum}).

\section{Proof of Theorem 4}
\label{app-C}
Define $z \triangleq \frac{{{r^2} + r_u^2}}{{\left| {{r^2} - r_u^2} \right|}}$. Notice that $z\ge 1$ always. Using \cite[8.882.1]{rvq:Table_of_Integrals}, we have
\setlength{\arraycolsep}{1pt}
\begin{eqnarray}\label{circle upper 1}
&&P_{\frac{v}{2} - 1}\left( \frac{{{r^2} + r_u^2}}{{\left| {{r^2} - r_u^2} \right|}} \right) ={P_{\frac{v}{2} - 1}}\left( z \right) \nonumber\\
&=& \frac{1}{\pi }\int_0^\pi  {{{\left( {z + \sqrt {{z^2} - 1} \cos \varphi } \right)}^{\frac{v}{2} - 1}}d\varphi } \notag\\
 &=& \frac{{{z^{\frac{v}{2} - 1}}}}{\pi }\int_0^\pi  {{{\left( {1 + \frac{{\sqrt {{z^2} - 1} }}{z}\cos \varphi } \right)}^{\frac{v}{2} - 1}}d\varphi
 }\nonumber \\
&=& \frac{{{z^{\frac{v}{2} - 1}}}}{\pi }\int_0^{\frac{\pi }{2}} {{{\left( {1 + {\frac{{\sqrt {{z^2} - 1} }}{z}}\cos \varphi } \right)}^{\frac{v}{2} - 1}}d\varphi }  + \frac{{{z^{\frac{v}{2} - 1}}}}{\pi }\int_{\frac{\pi }{2}}^\pi  {{{\left( {1 + {\frac{{\sqrt {{z^2} - 1} }}{z}}\cos \varphi } \right)}^{\frac{v}{2} - 1}}d\varphi } \notag \\
 &=& \frac{{{z^{\frac{v}{2} - 1}}}}{\pi }\int_0^{\frac{\pi }{2}} {\underbrace {\left[ {{{\left( {1 + {\frac{{\sqrt {{z^2} - 1} }}{z}}\cos \varphi } \right)}^{\frac{v}{2} - 1}} + {{\left( {1 - {\frac{{\sqrt {{z^2} - 1} }}{z}}\cos \varphi } \right)}^{\frac{v}{2} - 1}}} \right]}_{g\left( {v,z,\varphi } \right)}d\varphi}.
\label{49}
\end{eqnarray}
Next, we look for bounds for $g\left( {v,z,\varphi } \right)$. We derive the derivative of
${g\left( {v,z,\varphi } \right)}$ with respect to $z$ as follows:
\begin{align}\label{circle upper 3}
&\frac{\partial }{{\partial z}}g\left( {v,z,\varphi } \right) \notag\\
&= {\frac{{\frac{v}{2} - 1}}{{{z^2}\sqrt {{z^2} - 1} }} } \left[
{{{\left( {1 + {\frac{{\sqrt {{z^2} - 1} }}{z}}\cos \varphi }
\right)}^{{\frac{v}{2} - 2} }} - {{\left( {1 - {\frac{{\sqrt {{z^2}
- 1} }}{z}}\cos \varphi } \right)}^{{\frac{v}{2} - 2}}}} \right]\cos
\varphi.
\end{align}
Since $z\ge 1$, we have $\sqrt {{z^2} - 1}/z\in[0,1]$. Thus for  $\varphi\in[0,\frac{\pi }{2}]$,
\setlength{\arraycolsep}{5pt}
\[
\left\{\begin{array}{ll}
\frac{\partial }{{\partial z}}g\left( {v,z,\varphi } \right)\le 0 & \mbox{when $v\le 4 $} \\
\frac{\partial }{{\partial z}}g\left( {v,z,\varphi } \right)\ge 0 &
\mbox{when $v\ge 4 $}\end{array}\right..
\]
We can subsequently bound $g\left( {v,z,\varphi } \right)$ as follows:
\begin{equation}
\left\{\begin{array}{ll}
g(v,\infty,\varphi) \le g\left( {v,z,\varphi } \right)\le g(v,1,\varphi) & \mbox{when $2\le v\le 4 $} \\
g(v,1,\varphi) \le g\left( {v,z,\varphi } \right)\le g(v,\infty,\varphi)  & \mbox{when $6\ge v\ge 4 $}\end{array}\right..
 \label{bound-c}
\end{equation}
Define
\setlength{\arraycolsep}{1pt}
\begin{eqnarray}
B_1&\triangleq&\frac{{{z^{\frac{v}{2} - 1}}}}{\pi }\int_0^{\frac{\pi }{2}} g(v,1,\varphi)d\varphi
\nonumber\\
&=&\frac{{{z^{\frac{v}{2} - 1}}}}{\pi }\int_0^{\frac{\pi }{2}} {\left[ {{{\left( {1 + 0 \times \cos \varphi } \right)}^{\frac{v}{2} - 1}} + {{\left( {1 - 0 \times \cos \varphi } \right)}^{\frac{v}{2} - 1}}} \right]d\varphi }\notag = {z^{\frac{v}{2} - 1}}.\\
B_2&\triangleq&\frac{{{z^{\frac{v}{2} - 1}}}}{\pi }\int_0^{\frac{\pi }{2}} g(v,\infty,\varphi)d\varphi
\nonumber\\
&=&\frac{{{z^{\frac{v}{2} - 1}}}}{\pi }\int_0^{\frac{\pi }{2}} {\left[ {{{\left( {1 + \cos \varphi } \right)}^{\frac{v}{2} - 1}} + {{\left( {1 -  \cos \varphi } \right)}^{\frac{v}{2} - 1}}} \right]d\varphi }\nonumber\\
&&=\frac{{{z^{\frac{v}{2} - 1}}}}{\pi }\int_0^\pi  {{{\left( {1 + \cos \varphi } \right)}^{\frac{v}{2} - 1}}d\varphi } = \frac{{{z^{\frac{v}{2} - 1}}}}{\pi }\int_0^\pi  {{2^{\frac{v}{2}
- 1}}{{\cos }^{{{v} - 2}}}\left( {\frac{\varphi }{2}}
\right)d\varphi }\nonumber\\
&&=\frac{{{z^{\frac{v}{2} - 1}}}}{\pi }\int_0^{\frac{\pi }{2}} {{2^{{\frac{v}{2} }}}{{\cos }^{{{v} - 2}}}\left( t \right)dt}
= {\frac{{{2^{3\left(\frac{{v}}{2} - 1\right)}}\Gamma^2\left( {\frac{v}{2} - \frac{1}{2}} \right)}}{{\pi \Gamma \left( {v - 1} \right)}}{z^{\frac{v}{2} - 1}}}.\label{B2}
\end{eqnarray}
\setlength{\arraycolsep}{5pt}The last step is obtain by using
\cite[3.621.1]{rvq:Table_of_Integrals} and
\cite[8.384.1]{rvq:Table_of_Integrals}.

From (\ref{49}) and (\ref{bound-c}), $P_{\frac{v}{2} - 1}\left( z\right)$ can be bounded as:
\begin{equation}
\left\{\begin{array}{ll}
B_1 \ge P_{\frac{v}{2} - 1}\left( z\right)\ge B_2 & \mbox{when $2\le v\le 4 $} \\
B_1 \le P_{\frac{v}{2} - 1}\left( z\right)\le B_2  & \mbox{when $6\ge v\ge 4 $}\end{array}\right..
 \label{bound-d}
\end{equation}
By applying (\ref{bound-d}) in (\ref{capacity circle theorem 1}) and
(\ref{capacity circle theorem}), the first two lines of
(\ref{bound-R}) can be obtained. For the special cases of $\nu=2,4$,
with the aid of Eq. \cite[8.338.2]{rvq:Table_of_Integrals}, $\Gamma
\left( 1/2 \right) = \sqrt \pi  $, $\Gamma \left( {3/2} \right) =
\sqrt \pi/2  $, we have ${\frac{{{2^{\frac{3}{2}v - 3}}\Gamma^2
\left( {\frac{v}{2} - \frac{1}{2}} \right)}}{{\pi \Gamma \left( {v -
1} \right)}}}=1$. The two bounds are equal. Thus the last line of
(\ref{bound-R}) is proved.

%
%
%

\end{document}